\documentclass[10pt,journal,compsoc]{IEEEtran}
\usepackage{amsmath, amsthm, amsfonts, amssymb}
\usepackage{graphicx}
\usepackage{multirow}
\usepackage{cite}
\usepackage{enumerate}
\usepackage{subfigure}

\newtheorem{theorem}{\bf Theorem}

\newtheorem{lemma}[theorem]{\bf Lemma}

\newcounter{definition}

\newcommand{\dd}{\text{d}}

\long\def\symbolfootnote[#1]#2{\begingroup
\def\thefootnote{\fnsymbol{footnote}}\footnote[#1]{#2}\endgroup}

\begin{document}
\title{Optimal Slotted ALOHA under Delivery Deadline Constraint for Multiple-Packet Reception}

\author{Yijin Zhang, Yuan-Hsun Lo, Feng Shu, and Jun Li
\IEEEcompsocitemizethanks{\IEEEcompsocthanksitem Y. Zhang, F. Shu and J. Li are with the School of Electronic and Optical Engineering, Nanjing University of Science and Technology, Nanjing, China, and also with National Mobile Communications Research Laboratory, Southeast University, Nanjing, China. E-mails: yijin.zhang@gmail.com; \{shufeng, jun.li\}@njust.edu.cn.}
\IEEEcompsocitemizethanks{\IEEEcompsocthanksitem Y.-H. Lo is with the School of Mathematical Sciences, Xiamen University, Xiamen, China.
E-mail: yhlo0830@gmail.com.}
}

\IEEEtitleabstractindextext{
\begin{abstract}
This paper considers the slotted ALOHA protocol in a communication channel shared by $N$ users.
It is assumed that the channel has the multiple-packet reception (MPR) capability that allows the correct reception of up to $M$ ($1 \leq M < N$) time-overlapping packets.
To evaluate the reliability in the scenario that a packet needs to be transmitted within a strict delivery deadline $D$ ($D\geq 1$) (in unit of slot) since
its arrival at the head of queue, we consider the successful delivery probability (SDP) of a packet as performance metric of interest.
We derive the optimal transmission probability that maximizes the SDP for any $1\leq M<N$ and any $D\geq1$, and show it can be computed by a fixed-point iteration.
In particular, the case for $D=1$ (i.e., throughput maximization) is first completely addressed in this paper.
Based on these theoretical results, for real-life scenarios where $N$ may be unknown and changing, we develop a distributed algorithm that enables each user to tune its transmission probability at runtime according to the estimate of $N$.
Simulation results show that the proposed algorithm is effective in dynamic scenarios, with near-optimal performance.
\end{abstract}
}

\maketitle

\IEEEdisplaynontitleabstractindextext

\IEEEpeerreviewmaketitle

\section{Introduction}
\subsection{Motivation}
Since Abramson's seminal work~\cite{Abramson70} in 1970, ALOHA-type protocols have been widely used for initial terminal access or
short packet transmissions in a variety of distributed wireless communication systems due to their simplicity.
There were extensive studies on the slotted ALOHA under the traditional model of a {\em single-packet reception} (SPR) channel: a packet can be correctly received if there is no other packet transmission during its transmission.
However, the SPR has become restrictive due to the advent of multiple-packet reception (MPR) techniques that allow the correct reception of time-overlapping packets.
Hence, there is a natural interest in gaining a clear insight into the fundamental impact of MPR on the behavior of the slotted ALOHA protocol.

Differently from previous studies that dealt with the stability, throughput or delay issue of slotted ALOHA under MPR~\cite{Ghez88,Ghez89,Naware05,Luo06,Gau06,Gau08,ZZL09,ZZL10,BaeMPR14}, we in this paper concentrate on achieving maximum reliability in the scenario that a packet needs to be transmitted within a strict delivery deadline $D$ (in unit of slot) since its arrival at the head of queue.
Such a scenario can be safety message dissemination in vehicular networks~\cite{Stanica12} or machine to-machine communications in Internet of Things~\cite{IoT16}.
Some recent work on a similar issue can be found in~\cite{Bae13,Bae15J,Bae15O,Vinel12,Campolo12} for an SPR channel, and~\cite{Birk09,Birk99,Baron02b} for a multichannel system.

\subsection{Contribution}
Along the lines of~\cite{Gau06,Gau08,ZZL09,ZZL10,BaeMPR14}, this paper considers a specific MPR channel, namely the $M$-user MPR channel, in which up to $M$ time-overlapping packets can be received correctly.
Our key contributions are summarized as follows:
\begin{itemize}
  \item We derive the optimal transmission probability for the reliability maximization when $N$ users contend for the channel access.
Our work can be seen as a generalization of the work in~\cite{Bae13} that only focused on the SPR channel.
Moreover, we show that the optimal transmission probability can be computed by a fixed-point iteration.
  \item As explained in Section 3, the saturation throughput maximization of finite-user slotted ALOHA~\cite{ZZL09,BaeMPR14} can be studied as a special case $D=1$ of reliability maximization that we investigate here.
It should be pointed out that Bae et al.~\cite{BaeMPR14} obtained the optimal transmission probability for saturation throughput maximization necessarily relying on an unproved technical condition $\frac{\dd}{\dd\tau} \frac{\mathbb{E}[X^2]}{\mathbb{E}[X]} < N-1$ ($X$ is the number of users involved in each successful transmission and $\tau$ is the transmission probability).
In this paper, we present analysis to prove that this technical condition always holds for an arbitrary $1 \leq M<N$.
Hence, the issue of saturation throughput maximization under an $M$-user MPR channel is first completely addressed in this paper.
  \item Clearly, deriving the optimal transmission probability requires priori knowledge of $N$.
To achieve maximum reliability in real-life scenarios where $N$ is unknown and changing over time, built on the theoretical results derived
above, we propose a tuning algorithm that allows each user to estimate $N$ without requiring any access parameter input, and adjust its transmission probability accordingly at runtime.
Through an extensive performance study we show that the tuning algorithm is effective in a variety of dynamic network configurations considered in the paper, with near-optimal performance.
\end{itemize}

\subsection{Related work}
The first attempt to study slotted ALOHA under MPR was made by Ghez et al.~\cite{Ghez88,Ghez89}, in which they proposed the symmetric MPR channel model and analyzed stability properties under an infinite-user assumption.
Naware et al.~\cite{Naware05} extended the stability study to finite-user systems without posing any limitation on the MPR model, and in addition investigated the average delay in capture channels.
Luo et al.~\cite{Luo06} further established the throughput and stability regions for finite population over a standard MPR channel in which simultaneous packet transmissions are not helpful for the reception of any particular group of packets.

After the aforementioned studies for various generalized MPR channels, the throughput performance of slotted ALOHA over an $M$-user MPR channel has received much attention recently.
Gau~\cite{Gau06,Gau08} derived the saturation and non-saturation throughput for finite-user cases.
To demonstrate the capacity-enhancement, Zhang et al. in~\cite{ZZL09} proved that the maximum achievable throughput increases superlinearly with $M$ for both finite-user case with saturation traffic and infinite-user case with random traffic, and in~\cite{ZZL10} further showed that superlinear scaling also holds under bounded delay-moment requirements.
Following~\cite{ZZL09}, to fully utilize the $M$-user MPR channel, Bae et al.~\cite{BaeMPR14} derived the optimal transmission probability that maximizes the saturation throughput in the finite-user case under some unproved technical conditions.
To the best of our knowledge, little work has been done to investigate the reliability issue over an $M$-user MPR channel.
Our paper here is an attempt along this direction.

Under unknown and time-varying operating conditions, developing an estimation algorithm to acquire knowledge of the number of users $N$ is of significant importance.
Many approaches have been proposed for the SPR channel.
The method in~\cite{Rivest87} estimates $N$ based on the channel state in the previous slot, and the schemes in~\cite{Hajek82,Wu13} estimate $N$ with statistics of consecutive idle and collision slots.
However, all of them require that $N$ follows a Poisson distribution.
To remove the assumption on the distribution of $N$, the algorithm in~\cite{Bianchi03} applies an ARMA filter to estimate $N$ relying on the measured collision probability, the algorithm in~\cite{Cali00a} estimates $N$ by the number of idle slots between two successful transmissions, and the algorithm in~\cite{Cali00b} estimates $N$ from the knowledge of number of consecutive idle slots.
The extension of the estimation algorithm to the MPR case is rarely reported.
We are aware of only one previously proposed algorithm in~\cite{BaeMPR14} that estimates $N$ according to the collected information on the number of users involved in each successful transmission.
However, we find it is ineffective in some dynamic scenarios, which will be shown in Section 5.

The remainder of this paper is organized as follows.
In Section 2, we describe the considered system model.
In Section 3, we derive the optimal transmission probability that maximizes the reliability for any $1\leq M<N$ and any $D\geq1$, and show it can be computed by a fixed-point iteration.
In section 4, we propose a tuning algorithm by which each user can achieve a reliability level close to the theoretical limit at runtime.
In Section 5, simulation results verify the accuracy of our analytical results and the effectiveness of the proposed tuning algorithm.
Finally, Section 6 concludes this paper.

\section{System model}
As adopted in~\cite{Bae13,ZZL09,BaeMPR14}, we develop our analytical model based on the following assumptions:
\begin{enumerate}
  \item There are $N$ ($N\geq 2$) users with saturated traffic in the network, and all of them are within the transmission range of each other.
  \item The system is limited by user interference and the channel has an $M$-MPR capability, which means a packet can be correctly decoded by the receiver if at most $M-1$ other packet transmissions overlap with it at any time, and is unrecoverable otherwise. To avoid some trivial cases, we assume $1\leq M < N$. Specially, $M=1$ corresponds to the SPR channel.
  \item The channel time is divided into time slots of an equal length, and every packet exactly occupies the duration of one time slot.
  \item Each user knows the slot boundaries, and attempts to transmit a packet with a common transmission probability $\tau$ at the beginning of a time slot, $0 \leq \tau \leq 1$.
  \item Every packet is neither acknowledged nor retransmitted, since that an acknowledgement mechanism would incur extra overhead, waiting time and energy cost for short packets, and meanwhile, for some periodic traffic, the content can simply be replaced, and hence there is no need to retransmit an outdated packet.
  \item Every packet should be delivered within a strict delivery deadline $D$ ($D \geq 1$) (in unit of slot), which is defined as the duration from the moment of its arrival at the head of the queue to the completion of its transmission.
\end{enumerate}

\section{Optimal Transmission Probability}
Given any real number $\tau\in[0,1]$ and integer $D\geq 1$, let $P_D(\tau)$, called \emph{successful delivery probability} (SDP), be the probability that a packet will be successfully received within the delivery deadline $D$ under the common transmission probability $\tau$.
Consider a tagged user.
Let $Y$ denote the number of packets transmitted by the other $N-1$ users in a time slot.
It is easy to see $Y$ follows a binomial distribution with parameters $N-1$ and $\tau$, and then for $i=0,1,\ldots,N-1$, we have
\begin{equation}
\mathbb{P} (Y=i) = \binom{N-1}{i}  \tau^i (1-\tau)^{N-1-i}.
\label{eq:Y}
\end{equation}
Furthermore, the value $P_{D}(\tau)$ can be obtained as:
\begin{align}
P_{D}(\tau) &=\sum_{k=1}^D \tau (1-\tau)^{k-1} \mathbb{P} (Y \leq M-1) \notag \\
&=\sum_{k=1}^D \tau (1-\tau)^{k-1}  \sum_{i=0}^{M-1}  \binom{N-1}{i}  \tau^i (1-\tau)^{N-1-i} \notag \\
&= \big(1-(1-\tau)^D\big)  \sum_{i=0}^{M-1}  \binom{N-1}{i}  \tau^i (1-\tau)^{N-1-i}. \label{eq:P(tau)}
\end{align}

In this section, we aim to obtain the optimal transmission probability for maximizing $P_{D}(\tau)$.

For a given integer $D\geq 1$, let $\mathcal{P}^{max}_{D}$ denote the maximum SDP going through all possible $\tau\in[0,1]$, that is, $$\mathcal{P}^{max}_{D}:=\max_{\tau\in[0,1]}P_{D}(\tau).$$
Then, define the \emph{optimal transmission probability}, denoted by $\tau^{opt}_{D}$, to be the transmission probability such that the SDP achieves $\mathcal{P}^{max}_{D}$, i.e.,
\[
\tau^{opt}_D:=\arg\max_{\tau \in [0,1]} P_D(\tau).
\]
Note that $\tau^{opt}_{D}$ may not be unique by definition.


{\em Remark 1:} As $P_{1}(\tau)$ refers to the individual saturation throughput defined as the time average of the number of packets successfully transmitted by a user provided that all users have saturated traffic, $\tau^{opt}_{1}$ is indeed the optimal transmission probability maximizing the saturation throughput under MPR, which has been investigated in~\cite{BaeMPR14}.

{\em Remark 2:} When $M=1$, $\tau^{opt}_{D}$ is the optimal transmission probability maximizing the SDP within the delivery deadline $D$ under SPR, which has been derived in~\cite{Bae13}.

It is easy to see from \eqref{eq:P(tau)} that, when $D$ is fixed, $P_D(\tau)$ is a continuous function of $\tau$ on the closed interval $[0,1]$.
Hence, by The Extreme Value Theorem, $\tau^{opt}_{D}$ exists.
In the remainder of this subsection, we shall show the uniqueness of $\tau^{opt}_{D}$, and present how to obtain it.

Define the following semi open interval
\[
\mathcal{I}:=\Big[1-(\frac{N-1}{N-1+D})^\frac{1}{D}, 1\Big).
\]
We first provide some properties of $\tau^{opt}_{D}$.

\begin{lemma}\label{lemma:optimal-exist}
For any integers $D\geq 1$ and $1\leq M<N$,
\begin{enumerate}
  \item $\tau^{opt}_{D}$ is a solution of $\frac{d}{d\tau}P_{D}(\tau)=0$, and
  \item $\tau^{opt}_{D}$ must lie in $\mathcal{I}$.
\end{enumerate}
\end{lemma}
\begin{proof}


We prove these two statements by investigating the monotonicity of $P_D(\tau)$.
Define
\begin{equation}
f_1(\tau):=\sum_{i=0}^{M-1}  \binom{N-1}{i}  \tau^i (1-\tau)^{N-1-i}
\label{eq:f1}
\end{equation}
and
\begin{equation}
f_2(\tau):=\sum_{i=0}^{M-1}  i \binom{N-1}{i}  \tau^i (1-\tau)^{N-1-i}.
\label{eq:f2}
\end{equation}
Obviously, $f_1(\tau)>0$ and $f_2(\tau) \geq 0$ for $0<\tau<1$, $1\leq M < N$ and $D\geq1$.

By adopting the notation $f_1$ and $f_2$, the derivative of $P_{D}(\tau)$ with respect to $\tau$ can be written as
\begin{align}
\frac{\dd}{\dd\tau}&P_{D}(\tau)= D(1-\tau)^{D-1} f_1(\tau) \notag \\
 & \ \ \ \ + \Big (\frac{f_2(\tau)}{\tau(1-\tau)}-\frac{(N-1)f_1(\tau)}{1-\tau}\Big) \big(1-(1-\tau)^D\big) \notag \\
&=\Big((N+D-1)(1-\tau)^{D-1}-\frac{N-1}{1-\tau}\Big) f_1(\tau) \notag \\
& \ \ \ \ + \Big(\frac{1-(1-\tau)^D}{\tau(1-\tau)}\Big) f_2(\tau) \label{eq:P'1} \\
&=\frac{1}{\tau(1-\tau)} \Big(   \big(1-(1-\tau)^D\big)f_2(\tau)  \notag \\
& \ \ \ \ -\big(N-1-(N+D-1)(1-\tau)^D  \big)\tau f_1(\tau) \Big). \label{eq:P'2}
\end{align}
It is easy to see that $\frac{\dd}{\dd\tau}P_{D}(\tau)$ is continuous on the interval $(0,1)$.

Since that $P_D(\tau)>P_D(0)=P_D(1)=0$ if $\tau \in (0,1)$, we know the continuous function $P_D(\tau)$ has a local maximum at $\tau^{opt}_{D}$, which lies in $(0,1)$.
As $\frac{\dd}{\dd\tau}P_{D}(\tau)$ always exists on the interval $\tau \in (0,1)$, by the Fermat's Theorem, $\tau^{opt}_{D}$ is a solution of $\frac{\dd}{\dd\tau}P_{D}(\tau)=0$.

Furthermore, we have from \eqref{eq:P'1} that $\frac{\dd}{\dd\tau}P_{D}(\tau)>0$ for $\tau\in(0,1)\setminus\mathcal{I}$, and from \eqref{eq:P'2} that $\frac{\dd}{\dd\tau}P_{D}(\tau)<0$ as $\tau \rightarrow 1^-$.
By The Intermediate Value Theorem, the solutions of $\frac{\dd}{\dd\tau}P_{D}(\tau)=0$ must be in $\mathcal{I}$, i.e., $\tau^{opt}_{D}$ must lie in $\mathcal{I}$.
\end{proof}

Let
\[
H_1(\tau):= \frac{f_2(\tau)}{f_1(\tau)}\]
and
\[
H_2(\tau):= \tau \Big(N+D-1-\frac{D}{1-(1-\tau)^D}\Big).
\]
Following the proof of Lemma~\ref{lemma:optimal-exist}, in \eqref{eq:P'2}, $\tau^*$ is a solution of equation $\frac{\dd}{\dd\tau}P_{D}(\tau)=0$ if and only if it is a solution of the following equation:
\begin{equation}
H_1(\tau)-H_2(\tau)=0.
\label{eq:dd}
\end{equation}
In what follows, we will show that the equation \eqref{eq:dd} has a unique solution in the interval $\tau \in (0,1)$ by investigating the monotonicity of $H_1(\tau)$ and $H_2(\tau)$ separately.

\begin{lemma}\label{lemma:H1}
For $\tau\in(0,1)$, we have $\frac{\dd}{\dd\tau} H_1(\tau)=0$ if $M=1$, and otherwise
$$0 < \frac{\dd}{\dd\tau} H_1(\tau) < N-1.$$
\end{lemma}
\begin{proof}
The case for $M=1$ obviously holds as $H_1(\tau)=0$.
In the following, we only consider $2 \leq M < N$.

We first show that $\frac{\dd}{\dd\tau} H_1(\tau) > 0$ by a known result in~\cite{BaeMPR14}.
Let
\begin{equation}
T(\tau):=\frac{\sum_{i=1}^{M}  i^2 \binom{N}{i}  \tau^i (1-\tau)^{N-i}}{\sum_{i=1}^{M} i  \binom{N}{i}  \tau^i (1-\tau)^{N-i}}.
\label{eq:T}
\end{equation}
By letting $j=i-1$, after some algebraic manipulations, we have
\begin{align}
T(\tau)-1&=\frac{\sum_{i=1}^{M}  (i^2-i) \binom{N}{i}  \tau^i (1-\tau)^{N-i}}{\sum_{i=1}^{M} i  \binom{N}{i}  \tau^i (1-\tau)^{N-i}} \notag \\
&=\frac{\sum_{j=0}^{M-1}  j(j+1) \binom{N}{j+1}  \tau^{j+1} (1-\tau)^{N-1-j}}{\sum_{j=0}^{M-1} (j+1)  \binom{N}{j+1}  \tau^{j+1} (1-\tau)^{N-1-j}} \notag \\
&=\frac{\tau N\sum_{j=0}^{M-1}  j \binom{N-1}{j}  \tau^j (1-\tau)^{N-1-j}}{\tau N \sum_{j=0}^{M-1} \binom{N-1}{j}  \tau^{j} (1-\tau)^{N-1-j}} \notag \\
&=H_1(\tau). \label{eq:H1_1}
\end{align}
Since it has been proven in~\cite{BaeMPR14} that $\frac{\dd}{\dd\tau} T(\tau) \geq 0$ for $\tau\in(0,1)$ and $M>1$, we have $\frac{\dd}{\dd\tau}H_1(\tau) = \frac{\dd}{\dd\tau}T(\tau) > 0$ by \eqref{eq:H1_1}.

Now, we will show that $\frac{\dd}{\dd\tau} H_1(\tau)< N-1$.
By the binomial theorem, $H_1(\tau)$ can be rewritten as
\begin{align}
H_1(\tau)&=\frac{(N-1)\tau- \sum_{i=M}^{N-1} i \binom{N-1}{i}  \tau^i  (1-\tau)^{N-1-i}}{1-\sum_{i=M}^{N-1} \binom{N-1}{i}  \tau^i  (1-\tau)^{N-1-i} } \notag \\
&=(N-1)\tau \notag  \\
& \ \ \ \ -\frac{\sum_{i=M}^{N-1} \big(i-(N-1)\tau \big) \binom{N-1}{i}\tau^i(1-\tau)^{N-1-i} }{1-\sum_{i=M}^{N-1} \binom{N-1}{i}  \tau^i  (1-\tau)^{N-1-i}} \notag \\
&\stackrel{(*)}{=}(N-1)\tau -\frac{(N-M) \binom{N-1}{M-1} \tau^M (1-\tau)^{N-M}}{\sum_{i=0}^{M-1} \binom{N-1}{i}  \tau^i  (1-\tau)^{N-1-i} }, \notag \\
&= (N-1)\tau-(N-M) \binom{N-1}{M-1}R(\tau), \label{eq:H1_2}
\end{align}
where
\[
R(\tau):=\frac{\tau^{M} (1-\tau)^{N-M}}{\sum_{i=0}^{M-1} \binom{N-1}{i}  \tau^i  (1-\tau)^{N-1-i} }.
\]
The proof of $(*)$ is as follows.

For $m=2,3,\ldots,N-1$, we have
\begin{align*}
&  \binom{N-1}{m-1}(N-m)\tau^{m}(1-\tau)^{N-m} \\
& +\big(m-1-(N-1)\tau \big) \binom{N-1}{m-1}\tau^{m-1}(1-\tau)^{N-m} \\
=& \binom{N-1}{m-1}(N-m)\tau^{m}(1-\tau)^{N-m} \\
&- (N-m)  \binom{N-1}{m-1}\tau^{m-1}(1-\tau)^{N-m} \\
&+ (N-1)\binom{N-1}{m-1}\tau^{m-1}(1-\tau)^{N-m+1}\\
=& -\binom{N-1}{m-1}(N-m)\tau^{m-1}(1-\tau)^{N-m+1} \\
& + (N-1)\binom{N-1}{m-1}\tau^{m-1}(1-\tau)^{N-m+1}\\
=&(m-1)\binom{N-1}{m-1}\tau^{m-1}(1-\tau)^{N-m+1} \\
=& \binom{N-1}{m-2}(N-m+1)\tau^{m-1}(1-\tau)^{N-m+1}.
\end{align*}
Then by recursively using the above equation, we have
\begin{align*}
&\sum_{i=M}^{N-1} \big(i-(N-1)\tau \big) \binom{N-1}{i}\tau^i(1-\tau)^{N-1-i} \\
=&\binom{N-1}{N-2}\tau^{N-1}(1-\tau)\\
&+ \sum_{i=M}^{N-2} \big(i-(N-1)\tau \big) \binom{N-1}{i}\tau^i(1-\tau)^{N-1-i} \\
=&\binom{N-1}{M-1}(N-M)\tau^{M} (1-\tau)^{N-M}.
\end{align*}

To prove $\frac{\dd}{\dd\tau}H_1(\tau)<N-1$, by \eqref{eq:H1_2}, it suffices to show that $\frac{\dd}{\dd\tau}R(\tau)>0$.
Let
\[
Q(x) := \frac{x^{M-1}}{\sum_{i=0}^{M-1} \binom{N-1}{i} x^i}.
\]
By plugging $x=\frac{\tau}{1-\tau}$, we have
\begin{equation}
R(\tau)=\tau Q(x).
\label{eq:Qx}
\end{equation}
Note that as $\tau$ increases from $0$ to $1$, $x$ increases from 0 to $\infty$, and hence $Q(x)>0$.
By taking the derivative of $Q(x)$ with respect to $x$, we have, for $x>0$,
\begin{equation}
\frac{\dd}{\dd x}Q(x)=\frac{ \sum_{i=0}^{M-1}(M-1-i)\binom{N-1}{i} x^{M+i-2}}{\Big(\sum_{i=0}^{M-1} \binom{N-1}{i}  x^i\Big)^2} > 0.
\label{eq:dQ}
\end{equation}
Then,
\begin{equation}
\frac{\dd}{\dd\tau}R(\tau) = \frac{\dd}{\dd\tau} \Big(\tau Q(x)\Big) = Q(x) + \frac{\tau}{(1-\tau)^2}\cdot\frac{\dd}{\dd x}Q(x) > 0.
\label{eq:dR}
\end{equation}
Hence the result follows.
\end{proof}

\begin{lemma}\label{lemma:H2}
For $\tau\in(0,1)$, we have
$$\frac{\dd}{\dd\tau} H_2(\tau) > N-1.$$
\end{lemma}
\begin{proof}
Taking the derivative of $H_2(\tau)$ with respect to $\tau$ derives that
\begin{align*}
\frac{\dd}{\dd\tau}& H_2(\tau)=N+D-1-D\frac{1-(1-\tau)^D-D\tau(1-\tau)^{D-1}}{\big(1-(1-\tau)^D\big)^2}  \\
&=N-1+D\Big(1-\frac{1-(1-\tau)^{D}-D\tau(1-\tau)^{D-1}}{\big(1-(1-\tau)^D\big)^2}\Big)
\end{align*}

So we have
\begin{align}
&\frac{\dd}{\dd\tau}H_2(\tau) > N-1   \notag \\
\Leftrightarrow & \big(1-(1-\tau)^D\big)^2 > 1-(1-\tau)^{D}-D\tau(1-\tau)^{D-1} \notag \\
\Leftrightarrow & \big(1-(1-\tau)^D\big)^2 + (1-\tau)^{D}+D\tau(1-\tau)^{D-1}-1 > 0 \notag \\
\Leftrightarrow & (1-\tau)^{2D} - (1-\tau)^{D} + D\tau(1-\tau)^{D-1} > 0 \notag \\
\Leftrightarrow & (1-\tau)^{D+1} + (D+1)\tau -1 > 0 \label{eq:Gu}
\end{align}
Let $G(\tau):=(1-\tau)^{D+1} + (D+1)\tau -1$.
We have
\begin{align*}
\frac{\dd}{\dd\tau}G(\tau) &= -(D+1)(1-\tau)^D + D+1  \\
& = (D+1)\big(1-(1-\tau)^D \big),
\end{align*}
which is larger than $0$ for $\tau\in(0,1)$.
Therefore, $G(\tau)$ is strictly increasing for $\tau\in(0,1)$, which implies that $$G(\tau)>\lim_{\tau\to 0^+} G(\tau) =0.$$
Hence we complete the proof by \eqref{eq:Gu}.
\end{proof}

\begin{figure}
\begin{center}
  \includegraphics[height=2.5in,width=3.5in]{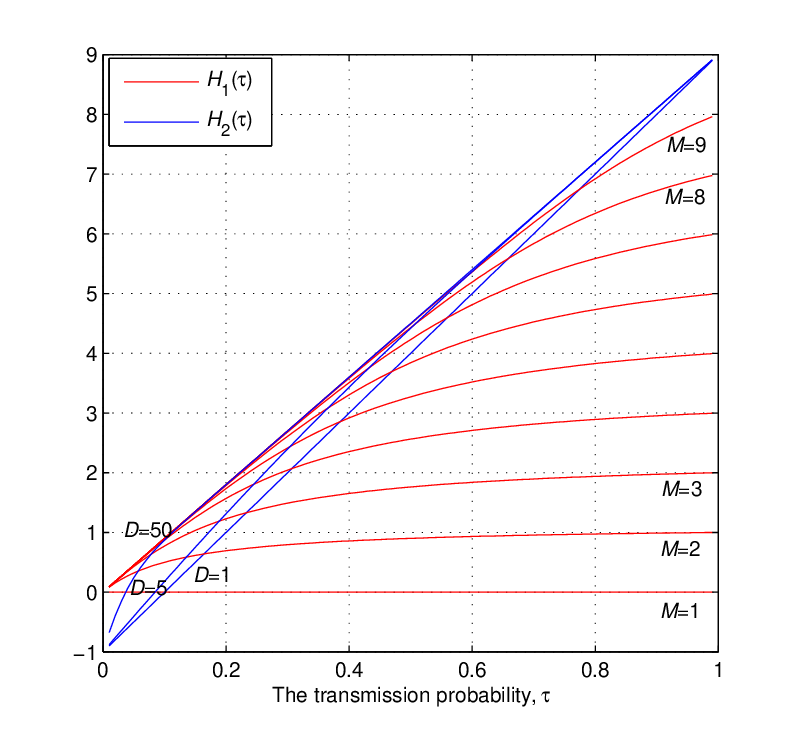}
\end{center}
\caption{$H_1(\tau)$ for the varying $M$ and $\tau$, $H_2(\tau)$ for the varying $D$ and $\tau$ when $N=10$.}
\label{fig:H12}
\end{figure}

To illustrate that $0 \leq \frac{\dd}{\dd\tau} H_1(\tau) < N-1 < \frac{\dd}{\dd\tau} H_2(\tau)$ on the interval $\tau \in (0,1)$ for any $D \geq 1$ and $1 \leq M < N$ obtained by Lemma 1 and Lemma 2, a numerical example is presented in Fig.~\ref{fig:H12}, which plots $H_1(\tau)$ and $H_2(\tau)$ for the varying $\tau$, $M$ and $D$ when $N=10$.

Now, we are ready to derive the uniqueness of $\tau^{opt}_{D}$.

\begin{theorem}\label{thm:optimal-unique}
For any integers $D\geq 1$ and $1\leq M<N$, the equation~\eqref{eq:dd} has a unique solution on $\tau$ in the interval $0 < \tau < 1$, denoted by $\tau^*$, and $\tau^{opt}_{D}=\tau^*$.
\end{theorem}
\begin{proof}
Suppose there are two distinct solutions to the equation~\eqref{eq:dd} in $(0,1)$.
By The Mean Value Theorem, there exists a solution of $\frac{\dd}{\dd\tau}H_1({\tau})=\frac{\dd}{\dd\tau}H_2({\tau})$.
However, by Lemma~\ref{lemma:H1} and Lemma~\ref{lemma:H2}, we have $$\frac{\dd}{\dd\tau}H_1(\tau) < N-1 < \frac{\dd}{\dd\tau}H_2(\tau), \quad \forall\tau\in(0,1). $$
This implies a contradiction to $\frac{\dd}{\dd\tau}H_1({\tau})=\frac{\dd}{\dd\tau}H_2({\tau})$.
Hence we conclude that the equation~\eqref{eq:dd} has a unique solution on $\tau$ in the interval $0 < \tau < 1$, which by Lemma~\ref{lemma:optimal-exist} promises the uniqueness of $\tau^{opt}_D$, and yields $\tau^{opt}_{D}=\tau^*$.
\end{proof}

{\em Remark 3:} In the context of saturation throughput maximization, Bae et al. in~\cite{BaeMPR14} derived the $\tau^{opt}_1$ under the assumption $\frac{d }{d\tau} T(\tau)<N$ in the interval $\tau \in (0,1)$.
They claimed that they proved $\frac{\dd}{\dd\tau} T(\tau)<N$ for $M=1 ,2, 3$, but could not prove it for any arbitrary $M$ due to extremely complex algebraic manipulations.
Here, the proof in Lemma~\ref{lemma:H1} has addressed this unsolved question, as $\frac{\dd}{\dd\tau} T(\tau)=\frac{\dd}{\dd\tau} H_1(\tau) <  N-1$.
In other words, the issue of saturation throughput maximization under an $M$-user MPR channel is first completely addressed in this paper.
Moreover, we in Theorem~\ref{thm:optimal-unique} obtained the existence and uniqueness of $\tau^{opt}_D$ for any $D \geq 1$ and any $1 \leq M < N$ without any assumption.

For the case $M=1$, which implies $H_1(\tau)=0$, it is easy to see from \eqref{eq:dd} that
\[
\tau^{opt}_{D}=1-\big(\frac{N-1}{N-1+D}\big)^\frac{1}{D}.
\]

For the case $1<M<N$, as $H_1(\tau)>H_2(\tau)=0$ when $\tau=1-(\frac{N-1}{N-1+D})^\frac{1}{D}$, we by Lemma~\ref{lemma:optimal-exist} know that
\[
\tau^{opt}_{D} \in \Big(1-\big(\frac{N-1}{N-1+D}\big)^\frac{1}{D}, 1\Big).
\]
Define a fixed-point iteration:
\[
x_{n+1}=x_n \cdot \frac{H_1(x_n)+1}{H_2(x_n)+1}
\]
for $n=0,1,2,\ldots$.
With the following theorem we guarantee that $\tau^{opt}_{D}$ for $M>1$ can be obtained by this fixed-point iteration.
\begin{theorem}\label{thm:FPI}
For any initial guess $x_0 \in (1-(\frac{N-1}{N-1+D})^\frac{1}{D}, 1)$, the sequence $x_0, x_1, x_2,\ldots$ converges to the fixed point $\tau^{opt}_{D}$ for $1<M<N$.
\end{theorem}
\begin{proof}
Define $$g(x):=\frac{x(H_1(x)+1)}{H_2(x)+1}$$ on the domain $(0,1)$.
By \eqref{eq:dd} and Theorem~\ref{thm:optimal-unique}, the equation $g(x)=x$ has a unique solution at $x=\tau^{opt}_{D}$ in $(0,1)$.
Therefore, the case that $x_0=\tau^{opt}_{D}$ obtains the fixed point $\tau^{opt}_{D}$ since $g(\tau^{opt}_{D})=\tau^{opt}_{D}$.
So we consider $x_0\neq\tau^{opt}_{D}$ in what follows.

Since $g(x)=x$ has a unique solution at $\tau^{opt}_{D}\in (0,1)$, to prove the sequence $\{x_n\}_0^\infty$ converges to $\tau^{opt}_D$, it suffices to show that
\begin{align}
\begin{cases}
x<g(x)<\tau^{opt}_D, & \text{for }x\in(1-(\frac{N-1}{N-1+D})^\frac{1}{D}, \tau^{opt}_D); \\
\tau^{opt}_D<g(x)<x, & \text{for }x\in(\tau^{opt}_D,1).
\end{cases}
\label{eq:converge}
\end{align}

As proved in Appendix that
\begin{equation}
\frac{\dd}{\dd x} g(x) > 0 \quad \text{for }x\in (0,1),
\label{eq:dg0}
\end{equation}
$g(x)$ is increasing on $(0,1)$.
This implies that
\begin{align}
\begin{cases}
g(x)<\tau^{opt}_D, & \text{for }x\in(1-(\frac{N-1}{N-1+D})^\frac{1}{D}, \tau^{opt}_D); \\
\tau^{opt}_D<g(x), & \text{for }x\in(\tau^{opt}_D,1).
\end{cases}
\label{eq:converge_1}
\end{align}
Since, in the interval $(0,1)$, $H_1(x)=H_2(x)$ only when $x=\tau^{opt}_D$ and $\frac{\dd}{\dd x}H_1(x) < \frac{\dd}{\dd x}H_2(x)$ by Lemma~\ref{lemma:H1} and Lemma~\ref{lemma:H2}, we have $0<H_1(x)<H_2(x)$ for $x\in(\tau^{opt}_D,1)$.
Then,
\begin{equation}
g(x)=x\frac{H_1(x)+1}{H_2(x)+1}<x, \quad \text{for } x\in(\tau^{opt}_D,1).
\label{eq:converge_2}
\end{equation}
Following the same argument, we have $H_1(x)>H_2(x)$ for $x\in(1-(\frac{N-1}{N-1+D})^\frac{1}{D}, \tau^{opt}_D)$.
By the result in Lemma~\ref{lemma:H2} that $\frac{\dd}{\dd x}H_2(x)>0$ and the L'Hospital's Rule that
$$\lim_{x\to 0^+}H_2(x)=-1,$$
we further have $H_1(x)>H_2(x)>-1$ for $x\in(1-(\frac{N-1}{N-1+D})^\frac{1}{D}, \tau^{opt}_D)$.
Then,
\begin{align}
g(x) =& x\frac{H_1(x)+1}{H_2(x)+1}>x, \notag \\
&\text{for } x\in(1-(\frac{N-1}{N-1+D})^\frac{1}{D}, \tau^{opt}_D).
\label{eq:converge_3}
\end{align}
Therefore, \eqref{eq:converge} can be derived combining \eqref{eq:converge_1}--\eqref{eq:converge_3}, and thus the result follows.
\end{proof}


\section{Runtime Optimization}
In the previous section, the optimal transmission probability has been derived, however, the analysis requires knowing in advance the number of users, $N$, which may be unavailable in some practical scenarios of distributed networks.
To cope with this restriction, we in this section develop a distributed algorithm to estimate $N$ for $M>1$, by which each user can tune the transmission probability to obtain an SDP close to the theoretical limit.

Consider a tagged user.
Recall the variable $Y$ that denotes the number of packets transmitted by the other $N-1$ users in a time slot.
By the distribution of $Y$ in~\eqref{eq:Y}, for $1 \leq i_1 < i_2 \leq M$, we obtain the following equation:
\begin{align}
&\frac{\mathbb{P} (Y=i_1)  \mathbb{P} (Y=i_2-1)}{\mathbb{P} (Y=i_2)\mathbb{P} (Y=i_1-1)} \notag  \\
&= \frac{\binom{N-1}{i_1} \binom{N-1}{i_2-1} \tau^{i_1+i_2-1} (1-\tau)^{2N-1-i_1-i_2}}{\binom{N-1}{i_2} \binom{N-1}{i_1-1} \tau^{i_1+i_2-1} (1-\tau)^{2N-1-i_1-i_2}} \notag \\
&= \frac{\binom{N-1}{i_1} \binom{N-1}{i_2-1}}{\binom{N-1}{i_2} \binom{N-1}{i_1-1}} =\frac{i_2(N-i_1)}{i_1(N-i_2)}.  \label{eq:PN}
\end{align}
It then directly follows that:
\begin{equation}
N=\frac{i_2(i_2-i_1)}{i_1\frac{\mathbb{P} (Y=i_1)  \mathbb{P} (Y=i_2-1)}{\mathbb{P} (Y=i_2)\mathbb{P} (Y=i_1-1)} -i_2}+i_2.
\label{eq:N}
\end{equation}
As $\frac{\mathbb{P} (Y=i_1)  \mathbb{P} (Y=i_2-1)}{\mathbb{P} (Y=i_2)\mathbb{P} (Y=i_1-1)}$ is locally measurable if a user is equipped with an array with at least $i_2$ antennas~\cite{Godara97,MPR_MAC13}, we find that~\eqref{eq:N} provides a linear function for the tagged user to estimate $N$ without priori knowledge of other access parameters.

We assume that the tagged user knows there are at most $N_{\text{max}}$ users in the network, but the actual number of users is changing and unknown.
An update interval is defined as a block of consecutive of $L$ time slots.
We require the tagged user to update the transmission probability at the beginning of the $n+1$th update interval, according to the following two necessary estimates of the network status.
\begin{enumerate}
  \item $\mu_n$: the estimated value of $\frac{\mathbb{P} (Y=i_1)  \mathbb{P} (Y=i_2-1)}{\mathbb{P} (Y=i_2)\mathbb{P} (Y=i_1-1)}$ at the end of the $n$th update interval;
  \item $N_n$: the estimated value of $N$ at the end of the $n$th update interval.
\end{enumerate}

The tagged user initially guesses that there are $N_0=N_{\text{max}}$ users, and sets $\mu_{0}=\frac{i_2(N_0-i_1)}{i_1(N_0-i_2)}$.
The transmission probability for the first update interval is obtained with given $N_0$, by Theorem~\ref{thm:optimal-unique} and a fixed-point iteration.
During the $n$th interval for $n=1,2,\ldots$, we describe the proposed algorithm as follows:
\begin{enumerate}
\item {\em Update $\mu_n$ at the end of the $n$th update interval}:
To evaluate $\frac{\mathbb{P} (Y=i_1)  \mathbb{P} (Y=i_2-1)}{\mathbb{P} (Y=i_2)\mathbb{P} (Y=i_1-1)}$ at runtime, the tagged user is required to record $\overline{A}_{i,n}$, the number of the slots during the $n$th update interval in which $i$ users are simultaneously transmitting and the tagged user is not transmitting, for $i=i_1-1,i_1,i_2-1$ and $i_2$.

Let $\overline{\mu}_n$ be the measure of $\frac{\mathbb{P} (Y=i_1)  \mathbb{P} (Y=i_2-1)}{\mathbb{P} (Y=i_2)\mathbb{P} (Y=i_1-1)}$ during the $n$th update interval, and its value is calculated by:
\[
\overline{\mu}_n =  \frac{\overline{A}_{i_1,n}\cdot \overline{A}_{i_2-1,n} }{\overline{A}_{i_2,n}\cdot \overline{A}_{i_1-1,n}}
\]

To avoid sharp changes in the estimated value, the tagged user further applies an Exponential Moving Average filter as follows:
\[
\mu_n=\delta\cdot\mu_{n-1}+ (1-\delta) \cdot \overline{\mu}_n
\]
where $\delta \in [0, 1]$ is a memory factor.
\item {\em Update $N_n$ at the end of the $n$th update interval}:
By~\eqref{eq:N}, the tagged user obtains
\[
N_n=\Big\langle \frac{i_2(i_2-i_1)}{i_1 {\mu}_n-i_2}+i_2 \Big\rangle
\]
where $\langle x \rangle$ is the integer closest to $x$.
\item {\em Update the transmission probability at the beginning of the $n+1$th update interval}: Update the transmission probability with given $N_n$, by Theorem~\ref{thm:optimal-unique} and a fixed-point iteration.
\end{enumerate}
To avoid harmful measure of $\overline{\mu}_{n}$ due to occasional fluctuations of the network status,
(i) if $\overline{\mu}_n$ is found to be smaller than $\mu_{0}$, a contradiction to the fact that~\eqref{eq:PN} must be equal to or larger than $\mu_0$ as $M<N \leq  N_0$, the tagged user sets $\overline{\mu}_n=\mu_{0}$; (ii) if $\overline{\mu}_n$ is found larger than $\frac{i_2(M+1-i_1)}{i_1(M+1-i_2)}$, a contradiction to~\eqref{eq:PN} as $N>M$, the tagged user sets $\overline{\mu}_{n}=\frac{i_2(M+1-i_1)}{i_1(M+1-i_2)}$; and
(iii) if $\overline{\mu}_{n}$ has an invalid value as $A_{i_2,n}\cdot A_{i_1-1,n}=0$, the tagged user sets $\overline{\mu}_n=\overline{\mu}_{n-1}$.

{\em Remark 4:}
It should be noted that, although the proposed algorithm for estimating $N$ is devised for the slotted ALOHA in saturation assumption, it can apply also to unsaturated traffic, provided that the estimation target becomes the average number of competing users rather than the total number of users; and can also apply to $\tau$-persistent CSMA, provided that the length of a slot may vary due to different channel status: idle, collision or success.

\section{Simulation Results}
In this section, to demonstrate the accuracy of the derived analytical results and the effectiveness of our proposed tuning algorithm, we present simulation results with respect to different network parameters obtained from Monte Carlo simulation developed in Matlab.
In the following, we will analyze the performance under stationary scenarios when $N$ is known, and then under dynamic scenarios when $N$ is unknown.
Each simulation point in stationary conditions represents the average value over 10 independent runs, each of which consists of $10^6$ slots.
The results in dynamic conditions are from a single representative simulation run for 500 updated intervals.

\subsection{Stationary Scenarios when $N$ is known}
\begin{figure*}
  \centering
  \subfigure[$M=2$]{
    \includegraphics[width=2.25in,height=2in]{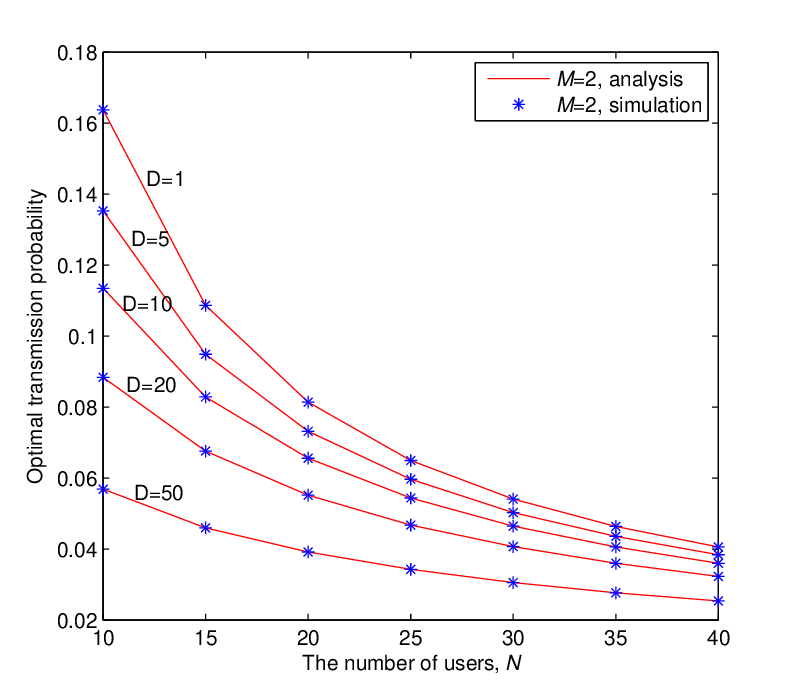}}
  \subfigure[$M=5$]{
    \includegraphics[width=2.25in,height=2in]{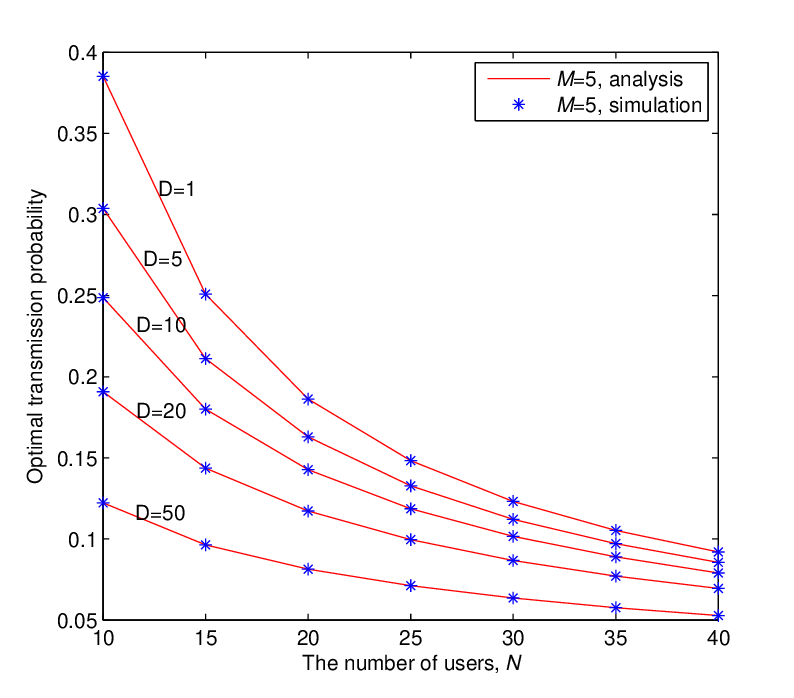}}
  \subfigure[$M=8$]{
    \includegraphics[width=2.25in,height=2in]{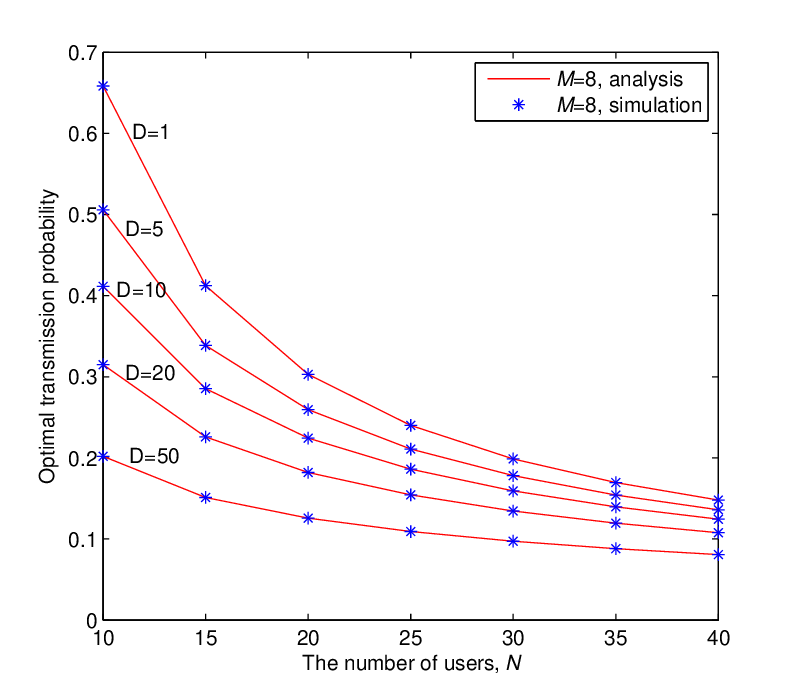}}
  \caption{The optimal transmission probability as a function of $N$ for different values of $M$ and $D$.}
  \label{fig:tau} 
\end{figure*}

\begin{figure*}
  \centering
  \subfigure[$M=2$]{
    \label{fig:subfig:a} 
    \includegraphics[width=2.25in,height=2in]{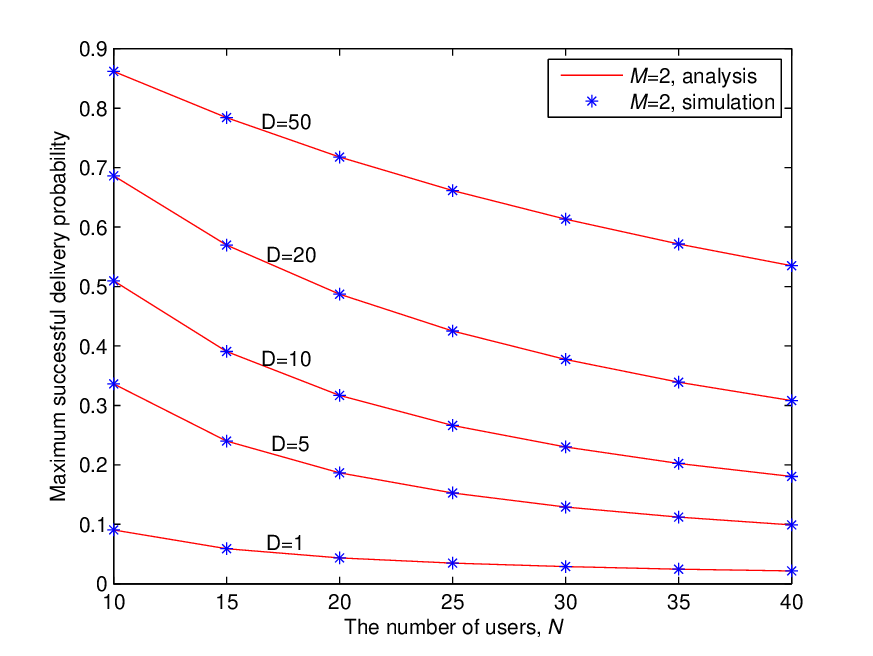}}
  \subfigure[$M=5$]{
    \label{fig:subfig:b} 
    \includegraphics[width=2.25in,height=2in]{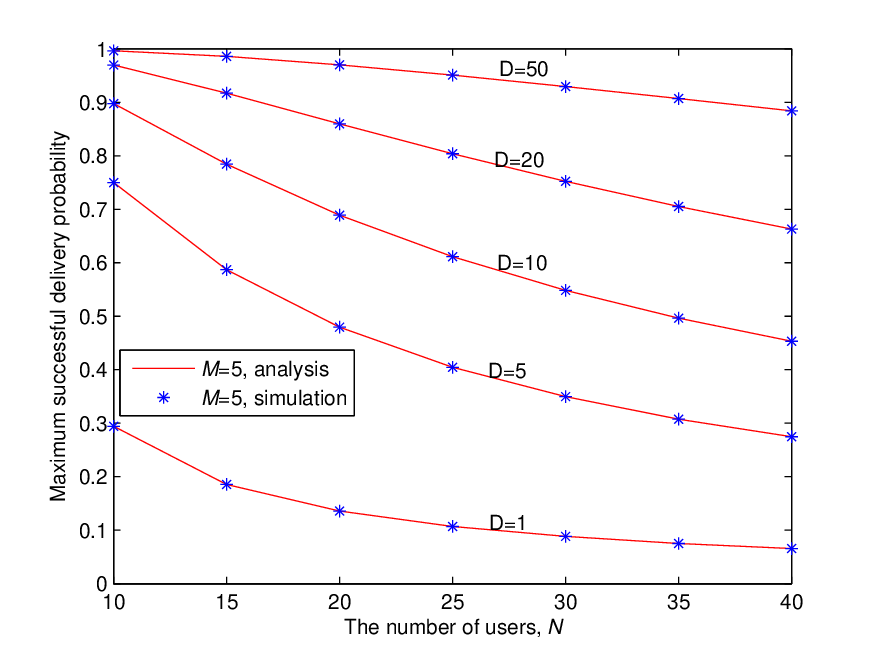}}
  \subfigure[$M=8$]{
    \label{fig:subfig:b} 
    \includegraphics[width=2.25in,height=2in]{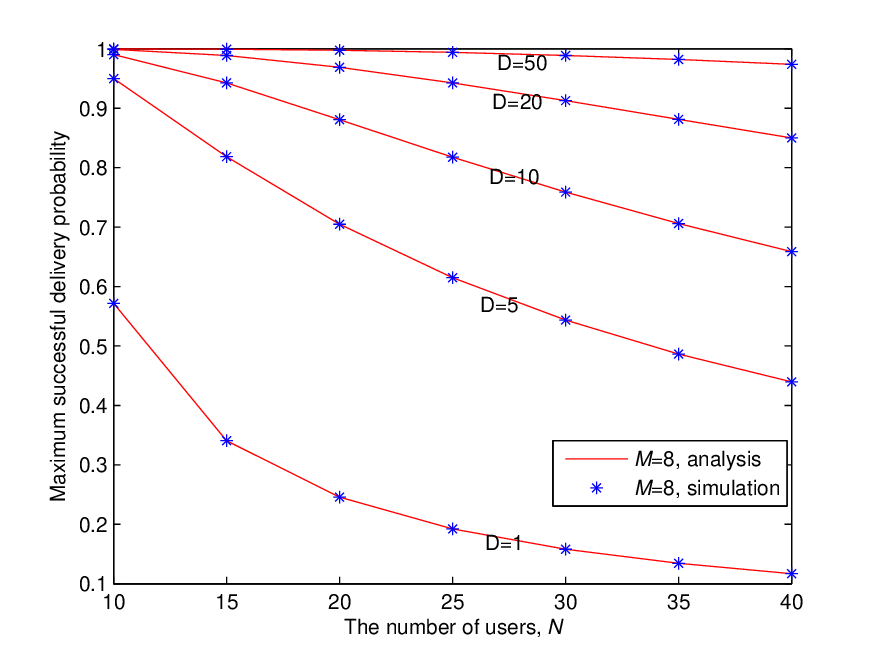}}
  \caption{The maximum SDP as a function of $N$ for different values of $M$ and $D$.}
  \label{fig:Ps} 
\end{figure*}
Fig.~\ref{fig:tau} and Fig.~\ref{fig:Ps} show the optimal transmission probability and the maximum SDP as a function of $N$ for different values of $M$ and $D$, respectively.
We see a good agreement between analytical and simulation results in all the scenarios.
Notice that the results for $D=1$ can be seen as the throughput maximization issue investigated in~\cite{BaeMPR14}.
In Fig.~\ref{fig:tau}, as expected, we observe that the optimal transmission probability becomes smaller when $N$ increases.
This is because that more users attempt to access the channel and the contention level becomes severer.
We also see the optimal transmission probability becomes larger when $M$ increases or $D$ decreases.
The reason is that a user needs to be more aggressive in accessing the channel if more concurrent packets can be successfully received or the user wants to successfully send out a packet within a shorter delivery deadline.
In Fig.~\ref{fig:Ps}, as expected, the curves show that a smaller $N$, a larger $M$ or a larger $D$ leads to a larger value of the maximum SDP.
In particular, we note that the optimal transmission probability for $D>1$ is smaller than that for $D=1$ for given $N$ and $M$.
This phenomenon indicates that the throughput maximization degrades the SDP for $D>1$, and in turn the maximization of SDP for $D>1$ degrades the throughput performance.

\begin{figure*}
  \centering
  \subfigure[$L=50000, \delta=0.7$]{
    \includegraphics[width=3.5in,height=2.5in]{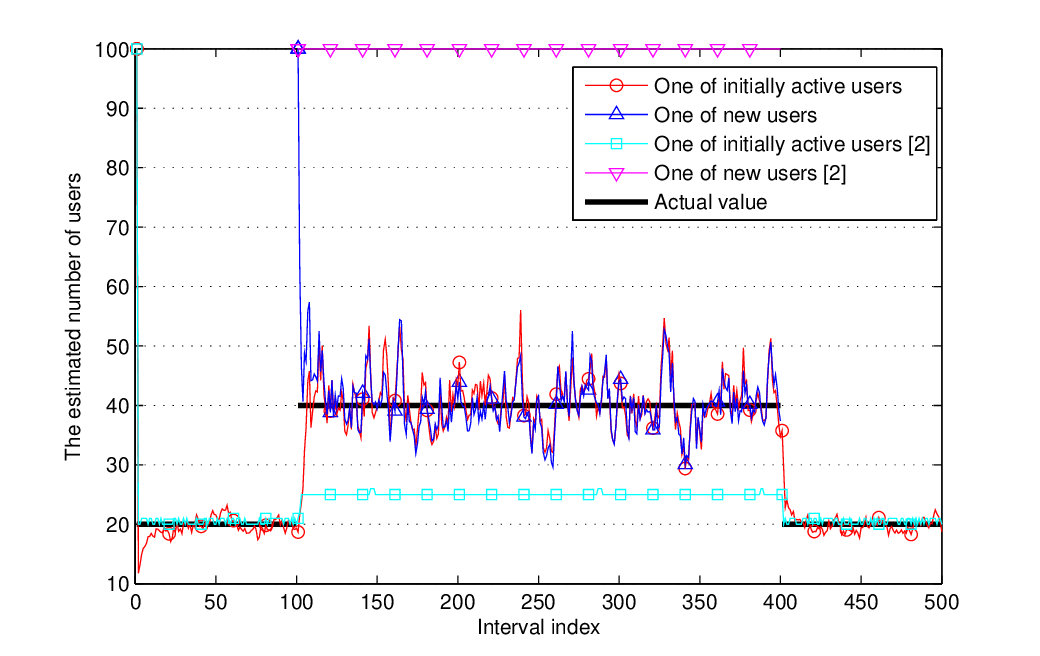}}
  \subfigure[$L=20000, \delta=0.9$]{
    \includegraphics[width=3.5in,height=2.5in]{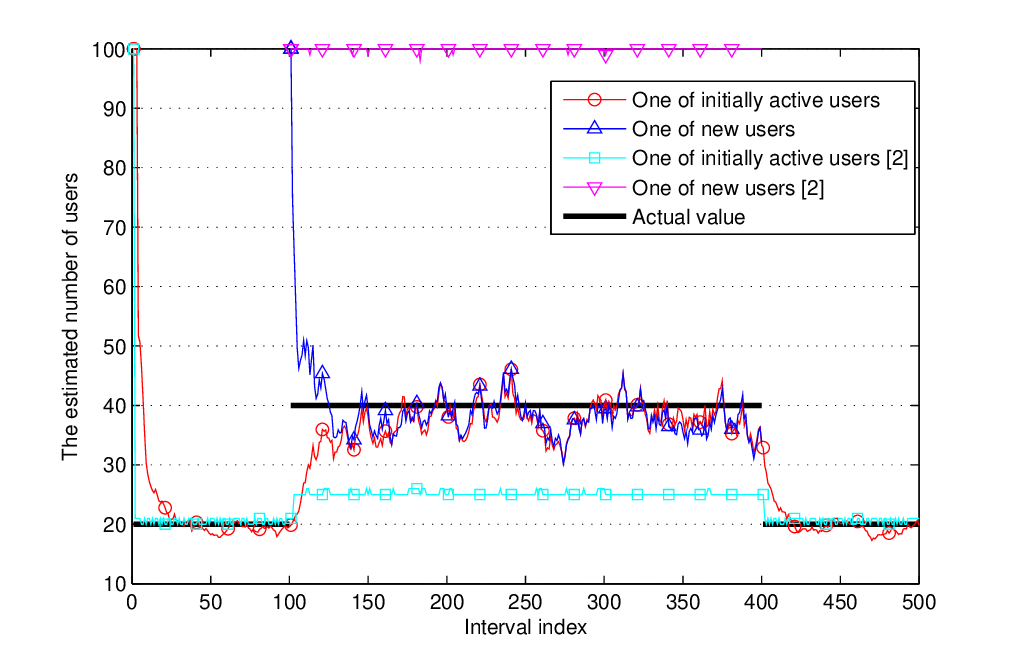}}
  \caption{The estimated number of users for $M=5$ and $D=1$ when the number of users sharply changes.}
  \label{fig:dynamicNe} 
\end{figure*}

\begin{figure*}
  \centering
  \subfigure[$L=50000, \delta=0.7$]{
    \includegraphics[width=3.5in,height=2.5in]{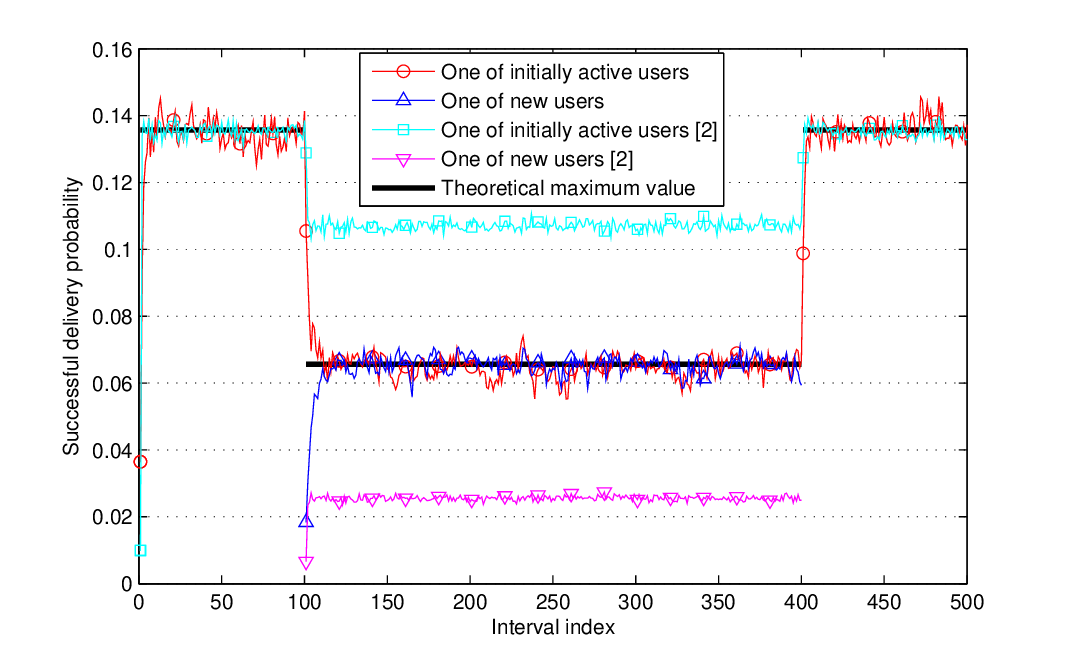}}
  \subfigure[$L=20000, \delta=0.9$]{
    \includegraphics[width=3.5in,height=2.5in]{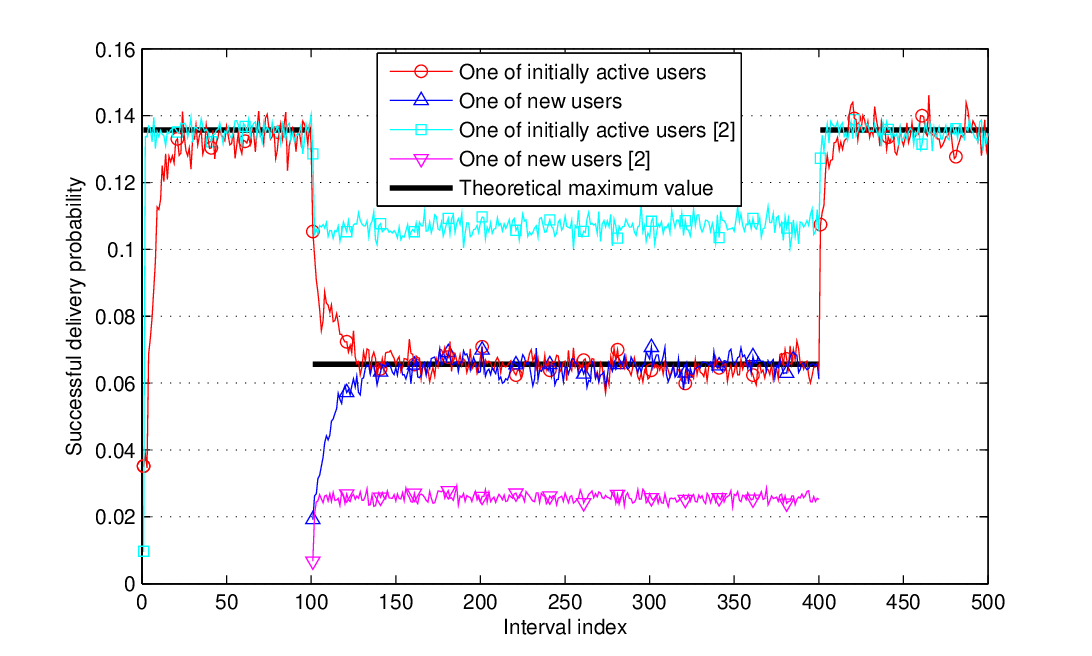}}
  \caption{Individual SDP for $M=5$ and $D=1$ when the number of users sharply changes.}
  \label{fig:dynamicPs} 
\end{figure*}

\begin{table*}
\[
\begin{array}{|c|c|c|c|c|} \hline
\text{Tuning parameters} & \text{Interval index} & \text{Theoretical maximum value} & \text{Mean in our simulation} &  \text{Variance in our simulation} \\ \hline
\multirow{3}{1.5cm}{$L=50000$, $\delta=0.7$}& 1-100  & 0.1357 & 0.1328 & 0.0012 \\ \cline{2-5}
& 101-400 & 0.0656 & 0.0646 & 5.54\cdot10^{-4} \\ \cline{2-5}
& 401-500 & 0.1357 & 0.1344   & 6.80\cdot10^{-4} \\ \cline{2-5} \hline
\multirow{3}{1.5cm}{$L=20000$, $\delta=0.9$}& 1-100 & 0.1357 & 0.1290 & 0.0027  \\ \cline{2-5}
& 101-400 & 0.0656 & 0.0647 & 0.0016\\ \cline{2-5}
& 401-500 & 0.1357 & 0.1337 & 0.0010  \\ \cline{2-5} \hline
\end{array}
\]
\caption{Mean and variance of individual SDP among all active users at different stages for $M=5$ and $D=1$.}
\label{table:ave}
\end{table*}

\subsection{Dynamic Scenarios when $N$ is unknown}
By setting $i_2=M$ and $i_1=2$ in~\eqref{eq:N}, we analyze the performance of the proposed tuning algorithm when the number of users sharply changes.
In detail, 20 users are always active from 1st to 500th updated interval, and 20 new users are active only from the 101st updated interval to the 400th updated interval.
Each user knows there are at most $N_{max}=100$ users, and initially guesses there are $N_0=100$ users.

\subsubsection{The cases with $D=1$}
We start by discussing the performance for $D=1$.
Fig.~\ref{fig:dynamicNe} and Fig.~\ref{fig:dynamicPs} show the estimated number of users and SDP for $M=5$ with different $L$ and $\delta$, respectively.
To characterize individual behavior and improve the readability, we only present the real trajectories of two representative users in the same run.

We observe from Fig.~\ref{fig:dynamicNe} that, in all the cases, when the actual number of users is changed to 20, the estimate rapidly adapts to changes and keeps a high level of accuracy; when the actual number of users is changed to 40, the estimate is still able to capture changes within a few update intervals, but the slop increases, i.e, the accuracy of the estimation degrades.
This phenomenon is because that as the actual number of users increases, the fluctuation in measuring $\overline{\mu}_n$ is amplified in estimation by using~\eqref{eq:N}.

As expected, we observe from Fig.~\ref{fig:dynamicPs} that the users achieve SDPs close to the theoretical limit at runtime by tuning the transmission probability according to the estimated number of users as shown in Fig.~\ref{fig:dynamicNe}.
Moreover, even when the oscillation of estimation is high, we find that the SDP still keeps relatively stable.
This phenomenon is due to the fact that the SDP is less sensitive to the change in transmission probability as the number of users increases.

To evaluate the fairness of the proposed algorithm, we further study mean and variance of SDP among all active users at different stages.
As shown in Table~\ref{table:ave}, we find the mean value is very close to the theoretical maximum value for each stage with a tolerance of at most 5\%, and the variance value is also very small.
This result indicates that the proposed algorithm can enable every active user to achieve near-optimal performance in dynamic scenarios.
Notice that the transient period is a dominant factor to degrade the performance.

As the case with $D=1$ can be viewed as runtime optimization of throughput, to show the proposed algorithm is more effective in a variety of dynamic scenarios, we consider the approach in~\cite{BaeMPR14} for comparison purposes.
The method therein recursively updates the transmission probability with an estimated $N$, and then uses the measure of $\frac{\mathbb{E}[X^2]}{\mathbb{E}[X]}$ ($X$ is the number of users involved in each successful transmission slot) to estimate new $N$ necessarily relying on the adopted transmission probability.
As shown in Fig.~\ref{fig:dynamicNe} and Fig.~\ref{fig:dynamicPs}, one sees their estimate keeps a very high level of accuracy when 20 new users are inactive, but is not able to keep track of the change well when 20 new users are active.
This phenomenon is because when two groups of users become active at different time instants, they may adopt different transmission probabilities, and then have biased estimate of $N$.
Whereas, our algorithm does not require any access parameter input, and hence would not lead to biased estimate in such a scenario.


\subsubsection{The cases with $D>1$}
Fig.~\ref{fig:dynamicNeD20} shows the real trajectories of the estimated number of users of two representative users for $M=5$ and $D=20$.
Similar to the case of $D=1$ as shown in Fig.~\ref{fig:dynamicNe}, the estimate is able to capture changes within a few update intervals, but the accuracy degrades when the actual number of users increases.

Fig.~\ref{fig:dynamicPsD20} shows the real trajectories of individual SDPs of two representative users for $M=5$ and $D=20$.
We observe that both users are able to adapt to different changes in the number of users, and achieve SDPs close to the
theoretical limit.
In addition, to evaluate the fairness issue, we in Table~\ref{table:aveD20} compare the mean and variance of individual SDP among all active users against the theoretical maximum SDP for $D=20$.
We confirm that the average performance of each user is near-optimal at every stage.

\begin{figure*}
  \centering
  \subfigure[$L=50000, \delta=0.7$]{
    \includegraphics[width=3.5in,height=2.5in]{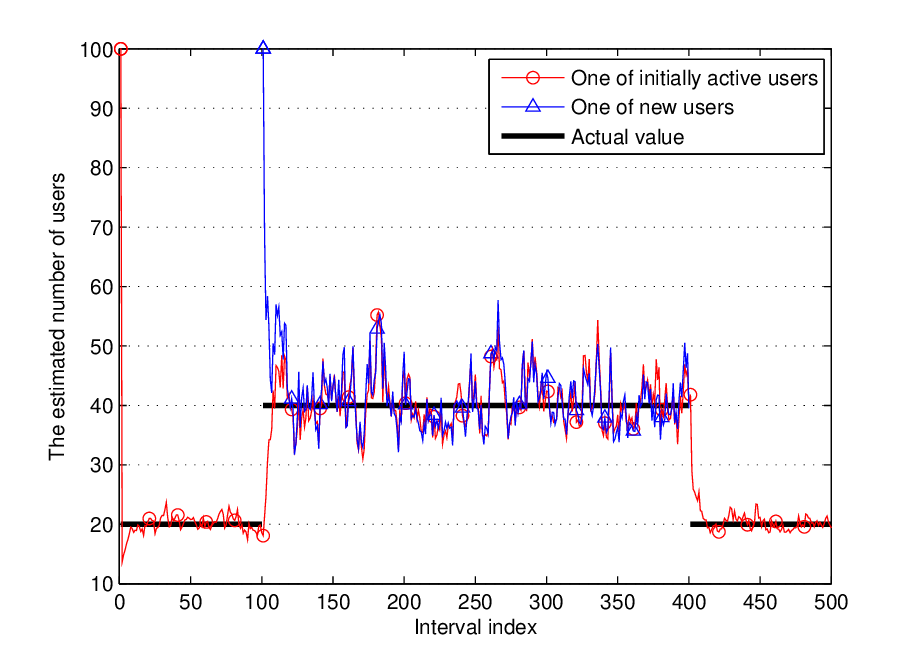}}
  \subfigure[$L=20000, \delta=0.9$]{
    \includegraphics[width=3.5in,height=2.5in]{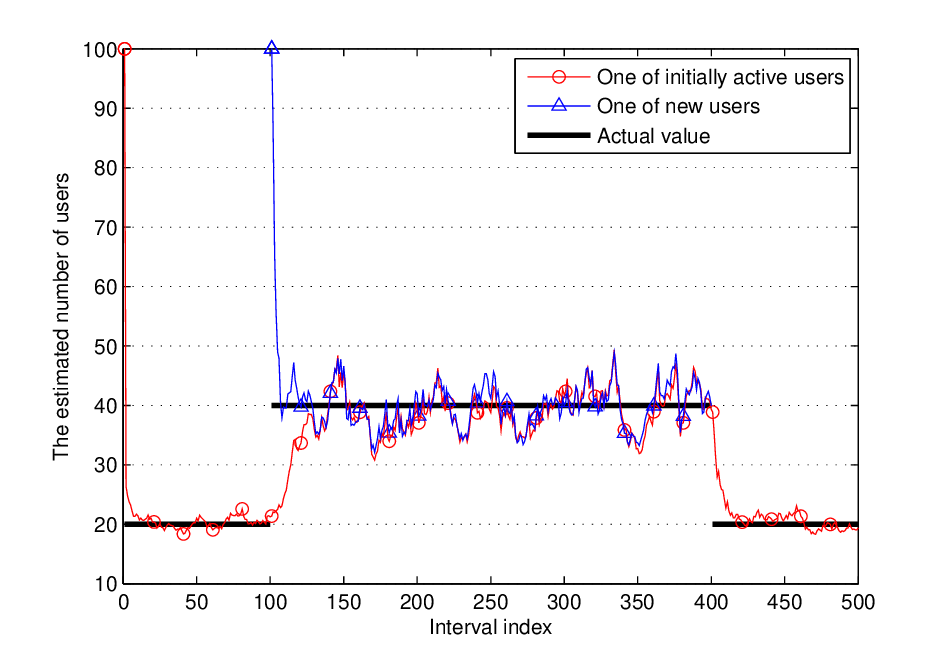}}
  \caption{The estimated number of users for $M=5$ and $D=20$ when the number of users sharply changes.}
  \label{fig:dynamicNeD20} 
\end{figure*}

\begin{figure*}
  \centering
  \subfigure[$L=50000, \delta=0.7$]{
    \includegraphics[width=3.5in,height=2.5in]{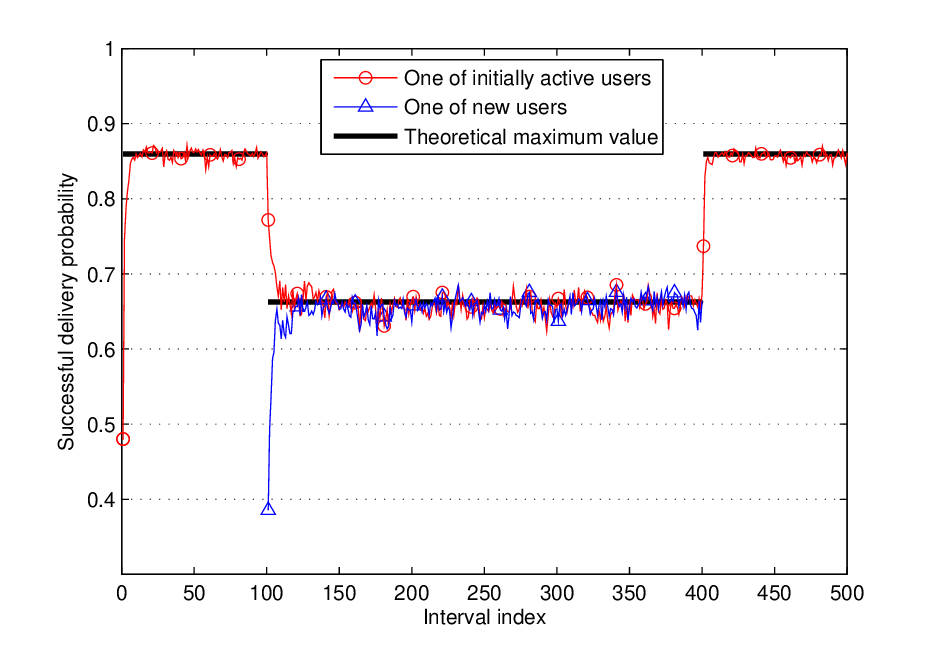}}
  \subfigure[$L=20000, \delta=0.9$]{
    \includegraphics[width=3.5in,height=2.5in]{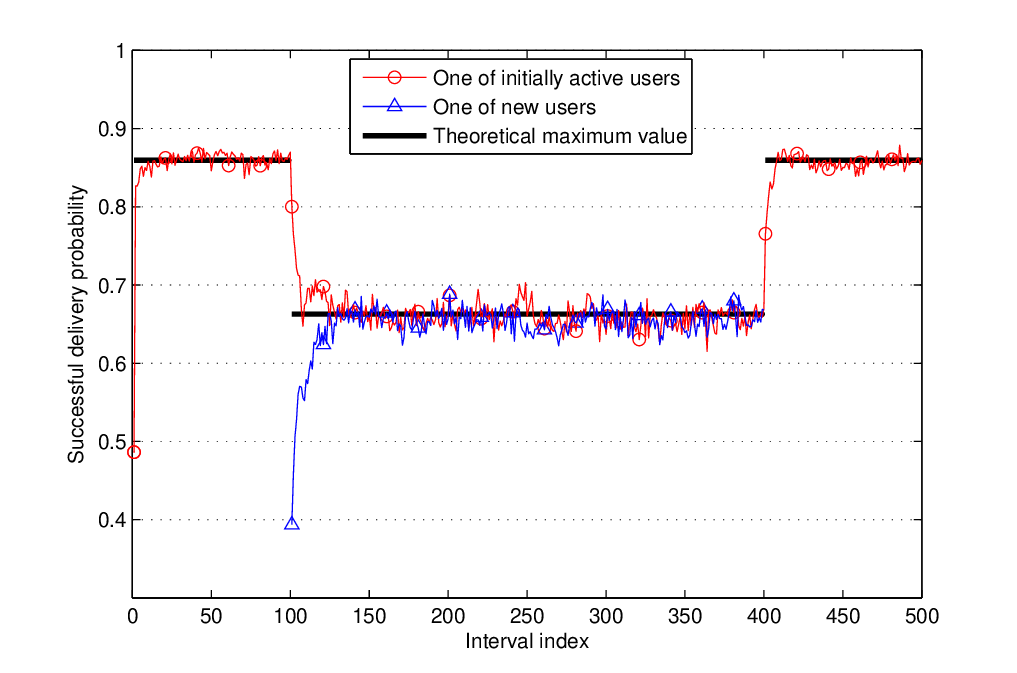}}
  \caption{Individual SDP for $M=5$ and $D=20$ when the number of users sharply changes.}
  \label{fig:dynamicPsD20} 
\end{figure*}

\begin{table*}
\[
\begin{array}{|c|c|c|c|c|} \hline
\text{Tuning parameters} & \text{Interval index} & \text{Theoretical maximum value} & \text{Mean in our simulation} &  \text{Variance in our simulation} \\ \hline
\multirow{3}{1.5cm}{$L=50000$, $\delta=0.7$}& 1-100  & 0.8595 & 0.8525  & 0.0017 \\ \cline{2-5}
& 101-400 & 0.6628 & 0.6578 & 0.0028 \\ \cline{2-5}
& 401-500 & 0.8595 & 0.8558 & 0.0011 \\ \cline{2-5} \hline
\multirow{3}{1.5cm}{$L=20000$, $\delta=0.9$}& 1-100 & 0.8595 & 0.8547 & 0.0053  \\ \cline{2-5}
& 101-400 & 0.6628 & 0.6579 & 0.0054\\ \cline{2-5}
& 401-500 & 0.8595 & 0.8545 & 0.0022  \\ \cline{2-5} \hline
\end{array}
\]
\caption{Mean and variance of individual SDP among all active users at different stages for $M=5$ and $D=20$.}
\label{table:aveD20}
\end{table*}

\section{Conclusion}
In this paper, we investigated the impact of MPR capability $M$ and delivery deadline $D$ on the optimal transmission probability that maximizes the SDP of the slotted ALOHA protocol in a communication channel shared by $N$ users with saturated traffic.
We have obtained the optimal transmission probability for any $1 \leq M <N$ and any $D\geq 1$, and shown it can be computed by a fixed-point iteration.
Then, we developed an adaptive tuning algorithm for maximizing the SDP when $N$ is unknown and time-varying.
Simulation results show the proposed algorithm is effective in a variety of dynamic scenarios.

As a special case of our work, the maximization of the SDP when $D=1$ can be viewed as the maximization of individual throughput.
It should be pointed out that \cite{BaeMPR14} is a first attempt to deal with this issue, however, the optimal transmission probability therein was obtained by necessarily assuming $\frac{\dd}{\dd\tau} T(\tau) < N-1$ for any $1 \leq M <N$ in the interval $0 < \tau < 1$.
Therefore, the throughput maximization for an $M$-user MPR channel is first completely addressed in this paper.

\appendix
\noindent \emph{Proof of \eqref{eq:dg0}.}
~~We first have
\begin{align*}
\frac{\dd}{\dd x}g(x) = x\frac{\dd H_1(x)}{\dd x} +& \frac{H_1(x)+1}{(H_2(x)+1)^2} \\
\cdot&\Big(1+H_2(x)+x\frac{\dd H_2(x)}{\dd x}\Big).
\end{align*}
Since, for $x\in(0,1)$, $\frac{\dd}{\dd x}H_1(x)>0$ by Lemma~\ref{lemma:H1} and $H_1(x)>0$ by definition, to prove $\frac{\dd}{\dd x}g(x)>0$ it suffices to show that
\begin{align}
1+H_2(x)+x\frac{\dd H_2(x)}{\dd x} =& 1-\frac{D^2x^2(1-x)^{D-1}}{(1-(1-x)^D)^2} \notag \\
 >& 0, \label{eq:dg0_suffi}
\end{align}
for $x\in(0,1)$.

Let $W(x):=\frac{D^2x^2(1-x)^{D-1}}{(1-(1-x)^D)^2}$.
Observe that
\begin{align*}
\frac{\dd}{\dd x}W(x) = \frac{2D^2x(1-x)^{D-2}}{(1-(1-x)^D)^3}\Big( &(1-Dx)\big(1-(1-x)^D\big) \notag \\
&- x(1-x)^D \Big).
\end{align*}
Since
\begin{align}
& ~1-Dx-(1-Dx+x)(1-x)^D \notag \\
=& ~1-Dx-(1-Dx+x)(1-x)^D \notag \\
<& ~1-Dx-(1-Dx+x) \label{eq:dg0_suffi_1} \\
=& ~-x < 0, \notag
\end{align}
where \eqref{eq:dg0_suffi_1} is due to $0<(1-x)^D<1$, we have $\frac{\dd}{\dd x}W(x)<0$ for $x\in(0,1)$, i.e., $W(x)$ is decreasing.
Moreover, by the L'Hospital's Rule we have
\begin{align*}
\lim_{x\to 0^+} W(x) = 1.
\end{align*}
This concludes that $W(x)<1$ for $x\in(0,1)$, which implies \eqref{eq:dg0_suffi}.
Hence we complete the proof.

\section*{Acknowledgements}
This work was partially supported by the National Natural Science Foundation of China (No. 11601454, 61472190, 61501238), the Open Research Fund of National Mobile Communications Research Laboratory, Southeast University (No. 2017D09, 2017D04), and the Natural Science Foundation of Fujian Province of China (No. 2016J05021).

\end{document}